\documentclass[11pt]{article}

\usepackage[english]{babel}
\usepackage[utf8]{inputenc}
\usepackage{lmodern}
\usepackage{graphicx}
\usepackage[colorinlistoftodos]{todonotes}
\usepackage{hyperref}

\usepackage{amsthm}
\usepackage{amsmath}
\usepackage{amssymb}
\usepackage{amsfonts}
\usepackage{mathrsfs}
\usepackage{mathtools}
\usepackage{verbatim}
\usepackage{footnote}
\usepackage[linesnumbered,boxed, vlined]{algorithm2e}
\usepackage{lineno}
\usepackage{caption}
\usepackage{tikzsymbols}
\usepackage[capitalize, nameinlink]{cleveref}
\usepackage{tikz}
\usepackage{boxedminipage2e}
\usepackage[section]{placeins}

\DeclareMathOperator{\E}{\mathbb{E}}

\crefname{theorem}{Theorem}{Theorems}
\Crefname{lemma}{Lemma}{Lemmas}
\Crefname{claim}{Claim}{Claims}
\Crefname{observation}{Observation}{Observations}

\newtheorem{theorem}{Theorem}
\newtheorem{lemma}[theorem]{Lemma}

\newtheorem*{remark}{Remark}

\renewcommand{\le}{\leqslant}
\renewcommand{\geq}{\geqslant}
\renewcommand{\leq}{\leqslant}

\newcommand{\EE}{\mathbb{E}}

\DeclareMathOperator{\poly}{poly}

\newcommand{\R}{{\mathbb{R}}}
\newcommand{\calM}{\mathcal{M}}
\newcommand{\calA}{\mathcal{A}}

\newcommand{\calE}{\mathcal{E}}

\newcommand{\ev}[1]{E_{#1}}
\newcommand{\lev}[1]{\calE_{#1}}
\newcommand{\allev}{E_{\operatorname{all}}}
\newcommand{\indpick}[1]{\mathbb{I}^{\operatorname{pick}}_{#1}}
\newcommand{\indquery}[1]{\mathbb{I}^{\operatorname{query}}_{#1}}

\newcommand{\mwbm}{\textsc{MWBM}\xspace}

\newcommand{\qc}{query\nobreakdash-\hspace{0pt}commit\xspace}
\newcommand{\adj}[1]{\delta(#1)}
\newcommand{\qcpt}[1]{P^{\operatorname{QC}}_{#1}}
\newcommand{\qcdist}[1]{\mathcal{D}^{\operatorname{QC}}_{#1}}
\newcommand{\qcutil}{\mathcal{U}}
\newcommand{\approxqc}{\texttt{Approx-QC}\xspace}

\newcommand{\qclp}{\operatorname{LP}_{\operatorname{QC}}}

\newcommand{\poiutil}{\mathcal{U}}
\newcommand{\poilp}{\operatorname{LP}_{\operatorname{PoI}}}
\newcommand{\poipt}[1]{P^{\operatorname{PoI}}_{#1}}
\newcommand{\poidist}[1]{\mathcal{D}^{\operatorname{PoI}}_{#1}}
\newcommand{\approxpoi}{\texttt{Approx-PoI}\xspace}

\newcommand{\bI}{\mathbb{I}}
\newcommand{\by}{\mathbf{y}}
\newcommand{\opt}{\operatorname{OPT}}

\begin{document}

\title{Beating Greedy for Stochastic Bipartite Matching}
 \author{
   Buddhima Gamlath \\
   EPFL\\
  \texttt{buddhima.gamlath@epfl.ch}
     \and
   Sagar Kale \\
  EPFL\\
   \texttt{sagar.kale@epfl.ch}
     \and
   Ola Svensson \\
  EPFL\\
   \texttt{ola.svensson@epfl.ch}
 }
\date{}

\maketitle

\begin{abstract}
  We consider the maximum bipartite matching problem in stochastic settings,
  namely the query-commit and price-of-information models.    In the
  query-commit model, an edge $e$ independently exists with probability $p_e$.
  We can query whether an edge exists or not, but if it does exist, then we
  \emph{have} to take it into our solution.  In the unweighted case, one can
  query edges in the order given by the classical online algorithm of Karp,
  Vazirani, and Vazirani~\cite{KarpVV1990} to get a $(1-1/e)$-approximation.
  In contrast, the previously best known algorithm in the \emph{weighted} case
  is the $(1/2)$-approximation achieved by the greedy algorithm that sorts the
  edges according to their weights and queries in that order.  
  
  Improving upon the basic greedy, we give a $(1-1/e)$-approximation algorithm
  in the weighted query-commit model.  We use a linear program (LP) to upper
  bound the optimum achieved by any strategy. The proposed LP admits several
  structural properties that play a crucial role in the design and analysis of
  our algorithm.  We also extend these techniques to get
  a $(1-1/e)$-approximation algorithm for maximum bipartite matching in the
  price-of-information model introduced by Singla~\cite{Singla2018}, who also
  used the basic greedy algorithm to give a $(1/2)$-approximation. 
\end{abstract}

\section{Introduction}
\label{sec:introduction}
Maximum matching is an important problem in theoretical computer science.  We
consider it in the query-commit and price-of-information models.  These settings
model the situation where the input is random, but we can know a specific part
of the input by incurring some cost, which is explicit in the
price-of-information setting, i.e., pay $\pi_e$ to know the weight of the edge
$e$, and implicit in the query-commit, i.e., the algorithm queries existence of
$e$, and if $e$ exists, the algorithm has to take $e$ into its solution.  We
formalize this now.  In the query-commit model, an edge $e$ independently exists
with probability $p_e$.  We can query whether an edge exists or not, but if it
does exist, then we \emph{have} to take it into our solution (hence the name
\emph{query-commit}).  So, specifically, for the matching problem, if $M$ is our
current matching, then we cannot query an edge $e$ if it intersects with $M$.
In the price-of-information model introduced by Singla~\cite{Singla2018}, edge
weights are random variables and we can query for the weight $W_e$ of an edge
$e$ by paying a (fixed) cost $\pi_e$.  We want to query a set $Q$ of edges and
output a matching $M \subseteq Q$ such that $W(M) - \sum_{e\in Q} \pi_e$ is
maximized. 

In the query-commit model, for bipartite graphs, one can query edges
in the order given by the classical algorithm of Karp, Vazirani, and
Vazirani~\cite{KarpVV1990} to get a $(1-1/e)$ approximation.  However, in the more
general weighted query-commit setting where an edge $e$ exists with weight
$w_e$ with probability $p_e$ and does not exist with probability $1-p_e$, it is
not clear how to use such a strategy. In fact, prior to our work, the best algorithm for the weighted setting was the basic greedy that sorts the edges by
weight $w_e$ and then queries in that order to get a $(1/2)$-approximate matching.  Similarly, in
the price-of-information setting, Singla gave a $(1/2)$-approximation algorithm
based on the greedy approach. 
In this work, we beat the greedy algorithm using
new techniques and give clean algorithms that achieve an approximation ratio of
$1-1/e$ (improving from $1/2$) in the settings of weighted query-commit as well
as price-of-information. A key component of our approach is to upper bound the optimum achieved by any strategy using a linear program (LP), and exploit several structural properties of this LP in the design of our algorithms. We now give a high level description of these techniques.

\subsection*{Techniques}
We first solve the weighted query-commit version, and this part of the paper
contains the essential ideas.  We then expand on these ideas to solve the
price-of-information version\footnote{Actually, we implicitly use
  ``multiple-weight query-commit'' as an intermediate step, where an edge has
  nonnegative random weight, and an algorithm can ask queries of the form ``is
  weight of $e$ greater than $c$,'' and if it is, then the algorithm has to add
  $e$ to its solution.  Then we use a reduction by Singla to reduce from
  price-of-information.}.  Let us first focus on the weighted query-commit
setting.  Here, the input is a bipartite graph $G = (A, B, E)$, and for each
$e \in E$, its existence probability $p_e$ and weight $w_e$.  Our goal is to
design a polynomial-time algorithm that gives a sequence of edges to query such
that after the last query, we end up with a matching of large weight.  First, to
get a handle on the expected value of an optimum strategy ($\opt$), consider the
linear program (LP) below.  For a vertex $u$, let $E_u$ be the set of edges
incident to $u$.
\begin{align*}
  &\text{Maximize } &&\sum_{e \in E} x_e \cdot w_e, \\
  &\text{subject to } &&  \sum_{e \in F} x_e &&\mkern-60mu\leq \Pr[\text{an edge
                                                in $F$ exists}]\,,
                      && \text{ for all } u \in A \cup B \text{ and for all } F \subseteq E_u\,,\\
  &  &   &\mkern35mux_e &&\mkern-60mu\geq 0 \,,&& \text{ for all } e \in E\,.
\end{align*}
Let $x'_e$ be the probability that $\opt$ solution contains $e$.  Since $\opt$
can have at most one edge incident to $u$, for any $F\subseteq E_u$, the events
in $\{\opt \text{contains }e : e\in F \}$ are disjoint.  Hence,
$\sum_{e \in F} x'_e$ is the probability that $\opt$ contains an edge in $F$,
which must be at most the probability that at least one edge in $F$ exists.
Therefore, $(x'_e)_{e\in E}$ has to satisfy the above LP.  We can solve this LP
in polynomial time using a subdmodular-function-minimization algorithm as a
separation oracle for the ellipsoid algorithm.  Let $\mathbf{x}^\ast$ be the
solution.  For a $u \in A$, if we restrict $\mathbf{x}^\ast$ to $E_u$ to get,
say, $\mathbf{x}^\ast_u$, we can write $\mathbf{x}^\ast_u$ as a convex
combination of extreme points of the polytope
\begin{equation}
  \label{eq:1}
  \left\{x \in \R^{E_u}_+ :  \sum_{e\in F} x_e \le \Pr[
  \text{an edge in $F$ exists}] \,\, \forall F\subseteq E_u\right\}\,.
\end{equation}
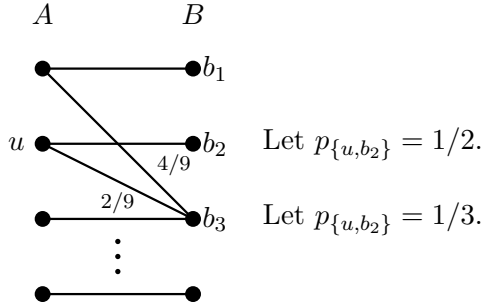
\begin{figure}
\centering
\begin{tikzpicture}[line width=0.3mm]

\draw(0,0) --  (2,0);
\draw(0,1) --  (2,1);
\draw(0,2) -- node[below] {\scriptsize $\mkern90mu4/9$} (2,2);
\draw(0,2) -- node[below] {\scriptsize $2/9$} (2,1);
\draw(0,3) --  (2,3);
\draw(0,3) --  (2,1);

\draw (0,4) node[anchor=north] {$A$};
\draw (2,4) node[anchor=north] {$B$};

\draw (-0.6,2) node[anchor=west] {$u$};
\draw (2.6,3) node[anchor=east] {$b_1$};
\draw (2.6,2) node[anchor=east] {$b_2$};
\draw (2.6,1) node[anchor=east] {$b_3$};

\draw (6,2) node[anchor=east] {Let $p_{\{u,b_2\}} = 1/2$.};
\draw (6,1) node[anchor=east] {Let $p_{\{u,b_2\}} = 1/3$.};

\fill [color=black] (0,0) circle (3pt);
\fill [color=black] (2,0) circle (3pt);
\fill [color=black] (0,1) circle (3pt);
\fill [color=black] (2,1) circle (3pt);
\fill [color=black] (0,2) circle (3pt);
\fill [color=black] (2,2) circle (3pt);
\fill [color=black] (0,3) circle (3pt);
\fill [color=black] (2,3) circle (3pt);

\fill [color=black] (1,0.3) circle (1pt);
\fill [color=black] (1,0.5) circle (1pt);
\fill [color=black] (1,0.7) circle (1pt);

\end{tikzpicture}
\caption{Here, we have $E_u = \{\{u,b_2\}, \{u,b_3\}\}$.  We write
  $\mathbf{x}^\ast_u = (4/9, 2/9)$ as a convex combination of extreme points of
  the polytope described in~\Cref{eq:1}:
  $\mathbf{x}^\ast_u = (2/3)(1/2, 1/6) + (1/3)(1/3, 1/3)$.  For $(1/2, 1/6)$,
  the inequalities corresponding to sets $\{\{u,b_2\}\}$ and
  $\{\{u,b_2\}, \{u,b_3\}\}$ are tight, and for $(1/3, 1/3)$, the inequalities
  corresponding to sets $\{\{u,b_3\}\}$ and $\{\{u,b_2\}, \{u,b_3\}\}$ are
  tight.  Observe that these form a chain.  We note that it is possible that an
  inequality corresponding to a nonnegativity constraint is tight.  Now, say we
  query the edges in the order given by the chain.  For $(1/2, 1/6)$, we first
  query $\{u,b_2\}$ then $\{u,b_3\}$, so we select $\{u,b_2\}$ with probability
  $p_{\{u,b_2\}}=1/2$ and $\{u,b_3\}$ with probability
  $(1-p_{\{u,b_2\}})p_{\{u,b_3\}} = 1/6$, which does indeed correspond to the
  extreme point $(1/2, 1/6)$.}\label{fig:convex}
\end{figure}
 
A key part of our approach is the nice structural properties of extreme points: 
the subsets corresponding to the tight constraints for an
extreme point $\by$ of the above polytope form a chain over a subset of $E_u$,
because the right hand side of the constraints is strictly submodular.  For a
$\by$, if we query the edges in the order given by its chain 
then it can be proved that we commit to an edge $e \in E_u$ with probability
$\by_e$.  Since $\mathbf{x}^\ast_u$ can be written as a convex combination of
such extreme points, if we select an extreme point with probability equal to its
coefficient in the convex combination and query the edges in the order given by
its chain, we commit to an edge $e \in E_u$ with probability
$\mathbf{x}^\ast_e$.  See~\Cref{fig:convex}.  But if we do this independently
for each vertex in $A$, then we end up with collisions on $B$, so we have to do
contention resolution there.  Effectively, for each vertex in $v\in B$, its
neighbor $u$ arrives independently with probability $\mathbf{x}^\ast_{uv}$ and
weight $w_{uv}$, and it has to pick one neighbor so that its expected utility is
close to $\sum_{uv}\mathbf{x}^\ast_{uv}\cdot w_{uv}$.  This setting is similar
to the prophet secretary problem, and we extend the ideas of Ehsani et
al.~\cite{Ehsani2018} to achieve this.

To generalize this from two-point distributions to multi-point distributions, we
think of an edge $e$ to have as many copies as the values its weight can take.
But now instead of being independent, the existence of these copies is correlated.
To handle this, we write a more general LP with constraints corresponding to
sets that are from a lattice family.  We can solve
submodular-function-minimization over a lattice family in polynomial
time~\cite{Grotschel1981}, which gives us a separation oracle for this LP.
Also, the extreme points of the polytope defined by constraints for a
left-hand-side vertex correspond to a chain with properties similar to the
two-point-distribution case (see~\Cref{lem:poi-chain}).

Once we have a query-commit algorithm for multi-point distributions, we can basically get
a price-of-information algorithm by a clean reduction~\cite{Singla2018}: for each edge
$e$ with weight given by the random variable $X_e$ and probing cost $\pi_e$, let
$\tau_e$ be the solution to the equation $\E[\max\{(X_e - \tau_e),0\}] = \pi_e$,
and let $Y_e = \min(X_e, \tau_e)$ be a new random variable.  Now run the
query-commit algorithm with $Y_e$, and whenever the algorithm queries any copy
of the edge $e$, we probe $e$, and  we only pay $\pi_e$ the first time we probe that edge. 

\subsubsection*{Organization of the Paper}
\label{sec:organization-paper}
Next, we review the related work.  In~\Cref{sec:qc}, we see the weighted
query-commit algorithm; we then extend it to the price-of-information setting
in~\Cref{sec:poi}.

\subsection*{Related Work}
\label{sec:related-work}

As mentioned earlier, the algorithm of Karp et al.~gives a $(1 - 1/e) = 0.632$
approximation in the unweighted query-commit model for bipartite graphs.
Costello et al.~\cite{CostelloTT2012} give a $0.573$-approximation for general
graphs and show that no algorithm can give an approximation better than $0.898$ compared to the optimal offline algorithm (that knows all the outcomes before selecting the matching).

Motivated by applications in kidney exchange and online dating, Chen et
al.~\cite{ChenIKMR09} consider the matching problem in the query-commit model
with further constraint that for each vertex $v$, the algorithm can query at
most $t_v$ edges incident to it ($t_v$ is a part of the input) and give a
$(1/4)$-approximation algorithm.  Bansal et al.~\cite{Bansal2012} improve it to
$(1/3)$ for bipartite graphs and to $(1/3.46)$ for general graphs, and also give
a $(1/4)$-approximation in the weighted query-commit (for general graphs); both
of these ratios for the unweighted case are further improved by Adamczyk et
al.~\cite{AdamczykGM15}, who give a $(1/3.709)$-approximation for general graphs.
Baveja et al.~\cite{Baveja2017} improve this to $1/3.224$.  We mention that this
setting is more general than the setting we consider, because $t_v = \deg(v)$
for us, i.e., we do not put any restriction on the number of edges incident to a
vertex that we can query.

Molinaro and Ravi~\cite{MolinaroR11} give an optimal algorithm for a very
special class of sparse graphs in the unweighted query-commit setting.

Blum et al.~\cite{Blum15} consider the maximum matching problem, where, in the
input graph, an edge $e$ exists with probability $p_e$, the algorithm can query
the existence of an edge and does not have to commit, but needs to minimize the
number of queries subject to outputting a good approximation.  This model is considered in several follow-up
works~\cite{Assadi16,Assadi17,Maehara18,Behnezhad18}.

\subsubsection*{Further Related Work in Other Stochastic Models}
Feldman et al.~\cite{Feldman09} consider an online variant of stochastic
matching where the algorithm gets as input a bipartite graph $G=(A,B,E)$, and a
distribution $\mathcal{D}$ over $B$, and $n$ elements are drawn i.i.d from $B$
according to $\mathcal{D}$ (so there may be repetitions) that the algorithm
accesses online.  When a copy $v \sim \mathcal{D}$ arrives online, we have to
match it to an unmatched vertex $u$ in $A$ such that $\{u,v\} \in E$.  Note that
here the existence of an edge is not random in itself but that of a vertex is.
Again, Karp et al.'s algorithm gives a $(1 - 1/e)$-approximation, and Feldman et
al.~give a $0.67$-approximation.  Significant amount of work has been done on this variant as well~\cite{Manshadi11,Haeupler11,Bansal2012,Jaillet13,AdamczykGM15,Brubach16}.

We would like to mention the works of Dean et al.~\cite{Dean04,Dean05}, where they consider
stochastic problems where cost of an input unit is only known as a probability
distribution that is instantiated after the algorithm commits to including the
item in the solution.  Charikar et al.~\cite{Charikar05} and Katriel et
al.~\cite{Katriel07} consider two stage optimization problems, where the first
stage is stochastic with a lower cost for decisions, and in the second stage,
with an increase in the decision cost, the actual input is known.


\section{A Weighted Query-Commit Algorithm}
\label{sec:qc}

In this section, we present a $(1 - 1/e)$-approximation algorithm for the 
maximum weight bipartite matching problem (\mwbm) in the \qc model. 

Let $G=(A, B, E, w)$ be a weighted bipartite graph with $|A| = |B| = n$,
where each edge $e \in E$ has weight $w(e)$ and exists independently with 
probability\footnote{Alternatively, one may think of the weight of an 
edge $e$ as an independent random variable that takes value $w(e)$ with 
probability $p_e$ and value zero with probability $1 - p_e$. 
We consider a more general distribution of edge weights in 
Section~\ref{sec:poi}.} $p_e$.
Let $\bI_e$ be the indicator variable for the event that edge $e$ exists 
(i.e., $\bI_e = 1$ if and only if edge $e$ exists).
The actual value of $\bI_e$ can only be determined by querying edge $e$ to 
check whether it exists or not. 
Given such a graph $G$, a \qc algorithm for \mwbm (adaptively) 
queries a sequence of edges $Q = (e_{q_1}, \dots, e_{q_m})$ and outputs a 
valid matching $M \subseteq Q$.
Since a \qc algorithm is \emph{committed} to add any queried edge 
that exists to its output, $M = \{ e \in Q: \bI_e = 1\}$.  

Let $\mathcal{A}$ be a \qc algorithm for \mwbm and let 
$M(\mathcal{A})$ denote its output on input $G$.
We define the expected utility of $\mathcal{A}$,
$\qcutil(\mathcal{A}) := \E [ \sum_{e \in M(\mathcal{A})} w(e) ]$
to be the expected weight of its output matching, where the expectation is 
over the randomness of the existence of edges and any internal randomness 
of the algorithm.
Let $\opt = \max_{\mathcal{A}} \qcutil(\mathcal{A})$, where the maximum is 
taken over all \qc algorithms\footnote{We remark that our approximation guarantee in the query-commit model also works with respect to the stronger offline adversary that knows all outcomes before selecting the matching. We have presented it in this way as the guarantees in the price-of-information model are of this type.}, be the optimal expected \qc utility.
We give a \qc algorithm \approxqc whose expected \qc utility
$\qcutil(\approxqc)$ is at least $(1 - 1/e) \opt$.

As described in the introduction, our approach consists of two stages. 
First, we solve a linear program to bound the optimal expected \qc 
utility $\opt$.
For each edge $e \in E$ we associate a variable $x_e$, which we think of as 
the probability that an optimal algorithm adds edge $e$ to its output. 
We set our LP constraints on $x_e$ variables to be those that must be 
satisfied by any algorithm that outputs a valid matching formed only 
of edges that exist.
Then the optimal expected utility is upper bounded by the maximum of the 
weighted sum $\sum_{e \in E} x_e \cdot w(e)$ subject to our LP constraints.

Then, we use the structural properties of our LP polytope to define
a distribution over the permutations of edges and use this distribution
to set the query order in our algorithm.
For this, we first define distributions $\qcdist{a}$ over the permutations 
of neighboring edges for each vertex $a \in A$ such that the following holds:
If we draw a random permutation $\sigma_a$ 
from $\qcdist{a}$ for each $a \in A$, query edges in the order of $\sigma_a$, and finally 
add to our output the first edge in $\sigma_a$ that exists, then the expected 
utility is optimal.
However, the output is not guaranteed to be a matching, because it might have
more than one edge incident to a vertex in $B$ (i.e., collisions).

To deal with the issue of collisions in $B$, we view the problem as 
a collection of \emph{prophet secretary} problem instances that has an 
instance for each vertex $b \in B$. 
In the prophet secretary problem, we have one item to sell (position to 
be filled) and a set of buyers (secretaries) arrive in a random order.
Each buyer makes a \emph{take it or leave it} offer to buy the item at some 
random price (or each secretary has some random skill level) whose 
distribution we know beforehand, and we are interested in maximizing the 
profit (or hiring the best secretary).

In their recent work,
Ehsani et al.~\cite{Ehsani2018} gave a $(1 - 1/e)$-competitive algorithm 
for the prophet secretary problem. 
Inspired by this result, we design our algorithm so that for each $b \in B$, 
it recovers at least $(1 - 1/e)$-fraction of the expected weight of the edge
incident to $b$ in an optimal algorithm's output. 
To elaborate, we pick a uniformly random permutation of the vertices in $A$ 
and treat them as buyers arriving at uniformly random times. 
Then, as seen from the perspective of a fixed vertex (i.e. an item) $b \in B$, 
buyers make a one-time offer with random values (corresponding to weights of edges incident to $b$), and
our algorithm uses dynamic thresholds that depend on arrival times to
make the decision of whether to sell the item at the offered price or not. 
One may view our algorithm as running the $(1 - 1/e)$-competitive 
prophet secretary algorithm in parallel for each of the vertices in $B$, but the
way we set the dynamic thresholds is different.

The rest of this section is organized as follows:
In Section~\ref{sec:qc-ub}, we describe how to upper bound $\opt$ using an LP.
In Section~\ref{sec:qc-lp-structure}, we analyze the structure of 
the LP polytope  and show how to construct the distributions $\qcdist{a}$ 
for each $a \in A$.
Finally in Section~\ref{sec:qc-algo}, we present our proposed 
algorithm and adapt the ideas of 
Ehsani et al.~\cite{Ehsani2018} to analyze its approximation guarantee.

\subsection{Upper-bounding the optimal expected utility}
\label{sec:qc-ub}

For each vertex $u \in A \cup B$, let $\adj{u}$ denote the set of edges 
incident to $u$.
For a subset of edges $F\subseteq E$, let $f(F)$ be the probability that at 
least one edge in $F$ exists, 
then $f(F) = 1 - \prod_{e \in F}(1 - p_e)$, because each edge $e$ exists
independently with probability $p_e$.

Fix any \qc algorithm $\mathcal{A}$.
For each edge $e \in E$, let $x_e$ be the probability that $\mathcal{A}$ 
outputs $e$. 
Consider a vertex $u \in A \cup B$ and a subset $F \subseteq \adj{u}$.
Since $\mathcal{A}$ outputs a valid matching, the events that edge $e$ being 
added to the output of $\mathcal{A}$ for each $e \in F$ are disjoint, and 
hence the probability that $\mathcal{A}$ adds one of the edges in $F$ to its 
output is $\sum_{e \in F} x_e$. 
But, for an edge to be added to the output, it must exist in the first place,
and thus it must be the case that $\sum_{e \in F} x_e \leq f(F)$ (because 
$f(F)$ is the probability that at least one edge in $F$ exists). 
Therefore, $\mathbf{x} = (x_e)_{e \in E}$ is a feasible solution to 
the following linear program, which we call $\qclp$.  So, the expected utility
$\qcutil(A)$ of $\mathcal{A}$ is at most the value of $\qclp$, which implies
Lemma~\ref{lem:qc-ub} below.
\begin{align*}
&\text{Maximize }   	& & \sum_{e \in E} x_e \cdot w(e), \\
&\text{subject to } 	& & \sum_{e \in F} x_e & & \mkern-100mu \leq f(F) \,,
                    	& & \text{ for all } u \in A \cup B \text{ and for all } 
                    			F \subseteq \adj{u} \,, \\
& & &\mkern35mux_e  	& & \mkern-100mu \geq 0 \,, 
	 									& & \text{ for all } e \in E\,.
\end{align*}

\begin{lemma} \label{lem:qc-ub}
The optimal expected \qc utility, $\opt$, is upper bounded by the 
value of $\qclp$.
\end{lemma}

\subsection{Structure of the LP and its implications}
\label{sec:qc-lp-structure}

Although $\qclp$ has exponentially many constraints, we can solve it in
polynomial time.

\begin{lemma} \label{lem:qc-solve-lp}
The linear program $\qclp$ is polynomial-time solvable.
\end{lemma}

\begin{proof}
Observe that for a fixed vertex $u \in A \cup B$, the constraints 
$\sum_{e \in F} x_e \leq f(F)$ for all $F \subseteq \adj{u}$, can be 
re-written as $0 \leq g_u(F)$ for all $F \subseteq \adj{u}$, where 
$g_u(F) = f(F) - \sum_{e \in F}x_e$ is a submodular function (notice that 
$f$ is submodular because it is a coverage function while $\sum_{e \in F}x_e$ 
is clearly modular).
Thus we can minimize $g_u$ over all subsets of $\adj{u}$ for all 
$u \in A \cup B$ in polynomial time using $O(n)$ submodular minimizations 
to find a violating constraint.
If none of the minimizations gives a negative value and if $x_e \geq 0$ 
for all $e \in E$, then the solution is feasible.
Thus we can solve $\qclp$ in polynomial-time using the ellipsoid method.
\end{proof}

For the rest of this section, we assume that $0 < p_e < 1$ for all $e \in E$.
We can safely ignore those edges $e \in E$ for which $p_e = 0$, 
and for those with $p_e = 1$, we can scale down the probabilities (at a small
loss in the objective value) due to the following lemma.

\begin{lemma} \label{lem:qc-scale-lp}
Let $\tilde{p}_e = (1 - \gamma) p_e$ for all $e \in E$, and for 
a subset $F \subseteq E$, let $\tilde{f}(F) = 1 - \prod_{e \in F}(1 - 
\tilde{p}_e)$  be the probability that at least one edge in $F$ exists 
under the scaled down probabilities $\tilde{p}_e$.
If we replace $f(F)$ in $\qclp$ by $\tilde{f}(F)$, the value of the resulting 
LP is at least $(1 - \gamma)$ times the value of $\qclp$.
\end{lemma}

\begin{proof}
Fix a set $F \subseteq E$ and label the edges in $F$ from $1$ through $|F|$. For $i = 1, \dots, |F|$, let $Q_i = 1 - \prod_{e = 1}^i(1 - p_e)$ and let $\tilde{Q}_i = 1 - \prod_{e = 1}^i(1 - \tilde{p}_e)$. Then, for $i > 1$ we have that
$Q_i = Q_{i-1} + p_i(1 - Q_{i-1})$, and similarly, $\tilde{Q}_i = \tilde{Q}_{i-1} + \tilde{p}_i(1 - \tilde{Q}_{i-1})$. By definition, we have $f(F) = Q_{|F|}$ and $\tilde{f}(F) = \tilde{Q}_{|F|}$.
We now prove that $\tilde{Q}_{i} \geq (1 - \gamma) Q_{i}$ for $i = 1, \dots, |F|$ by induction. 

For the base case, we have $\tilde{Q}_{1} = (1 - \gamma)Q_{1}$. Notice that, by the definition of $\tilde{p}_i$'s, we have $Q_i \geq \tilde{Q}_i$. Thus, for $i > 1$ we have that
\begin{align*}
    \tilde{Q}_i &= \tilde{Q}_{i-1} + \tilde{p}_i(1 - \tilde{Q}_{i-1})  \\
    &\geq (1 - \gamma) Q_{i-1} + (1 - \gamma)p_i(1 - \tilde{Q}_{i-1}) & & & \text{(by inductive hypothesis)}\\
    &\geq(1 - \gamma) Q_{i-1} + (1 - \gamma)p_i(1 - Q_{i-1}) & & & \text{(because $Q_{i-1} \geq \tilde{Q}_{i-1}$)} \\
    &= (1 - \gamma)(Q_{i-1} + p_i(1 - Q_{i-1})) = (1 - \gamma) Q_i.
\end{align*}
	
Thus if we scale down the polytope defined by the constraints of 
$\qclp$ by a factor of $(1 - \gamma)$, the resulting polytope is contained 
inside the polytope defined by $\tilde{f}(F)$ constraints. 
Moreover, all extreme points of both the polytopes have non-negative 
coordinates and the objective function has non-negative coefficients.
Hence the claim of Lemma~\ref{lem:qc-scale-lp} follows.
\end{proof}

\begin{remark}
The expected utility of our proposed algorithm is 
$(1 - 1/e) \cdot \opt^\ast \geq (1 - 1/e) \cdot \opt$, where $\opt^\ast$ 
is the optimal LP value of $\qclp$.
Thus, in the cases where the assumption $p_e < 1$ for all $e \in E$ does 
not hold, we can scale the probabilities down by $(1 - \gamma)$, and 
consequently the guarantee on expected utility will at least be 
$(1 - \gamma)(1 - 1/e) \cdot \opt$ due to Lemma~\ref{lem:qc-scale-lp}. 
We can choose $\gamma$ to be arbitrarily small.
To implement the scaling down operation, we can simply replace
each query made by an algorithm with a function that only queries with 
probability $(1 - \gamma)$.
\end{remark}

The assumption that $0< p_e < 1$ for all $e \in E$ yields the following lemma 
on the function $f$.

\begin{lemma} \label{lem:qc-strict-submod}
Fix a vertex $u \in A \cup B$, and suppose that $0< p_e < 1$ for all 
$e \in \adj{u}$. 
Then the function $f$ is strictly submodular 
and strictly increasing on subsets of $\adj{u}$.
That is:
\begin{enumerate}
\item \label{prop:qc-strict-submod} For all subsets $A, B \subseteq \adj{u}$ 
such that $A \setminus B \neq \emptyset$ and $B \setminus A \neq \emptyset$, 
$f(A) + f(B) > f(A \cup B) + f(A \cap B)$.
\item \label{prop:qc-strict-inc} For all $A \subsetneq B \subseteq 
\adj{u}$, $f(A) < f(B)$.
\end{enumerate}
\end{lemma}

\begin{proof}
Let $g(F) = 1 - f(F) = \prod_{e \in F}(1 - p_e)$ (note that 
$g(\emptyset) = 1$).
Notice that for $F_1, F_2 \subseteq F$ such that $F_1 \cap F_2 = \emptyset$ 
and $F_1 \cup F_2 = F$, it holds that $g(F) = g(F_1) \cdot g(F_2)$.
Let $A, B \subseteq \adj{u}$ be two sets such that $A \setminus B \neq 
\emptyset$ and $B \setminus A \neq \emptyset$.
It is sufficient to show that $g(A) + g(B) < g(A \cup B) + g(A \cap B)$.
We have 
\begin{align}
g(A) + g(B) &= g(A \cap B)\left( \underbrace{g(A \setminus B)}_{a} + 
\underbrace{g(B \setminus A)}_{b} \right), \label{eq:a1}
\intertext{and }
g(A \cup B) + g(A \cap B) &= g(A \cap B) \left( \underbrace{g(A \setminus B) 
\cdot g(B \setminus A)}_{a \cdot b} + 1 \right). \label{eq:a2}
\end{align}
Since $A \setminus B \neq \emptyset$ and $B \setminus A \neq \emptyset$, 
and $0 < p_e < 1$ for all $e \in E$, we have $g(A \cap B) > 0$ and both 
$a, b < 1$. 
Thus $a + b < 1 + a \cdot b$, because $(1 - a)(1 - b) > 0$. 
This combined with Equations \eqref{eq:a1} and \eqref{eq:a2} yields 
Property~\ref{prop:qc-strict-submod}.

Now consider $A \subsetneq \adj{u}$ and any edge $e \in \adj{u} \setminus A$. 
To prove Property~\ref{prop:qc-strict-inc}, it is sufficient to show that
$f(A \cup \{ e \}) > f(A)$, or equivalently, $g(A \cup \{ e \}) < g(A)$.
This is straightforward since $g(A \cup \{e\})/{g(A)} = 1 - p_e < 1$,
because $p_e > 0$.
\end{proof}

Let $\mathbf{x}^\ast = (x^\ast_e)_{e \in E}$ be an optimal solution to 
$\qclp$.  
Fix a vertex $a \in A$.  
Then, $\mathbf{x}_{a} = (x^\ast_e)_{e \in \adj{a}}$, which is 
$\mathbf{x}^\ast$ restricted only to those coordinates that correspond to 
edges in $\adj{a}$, satisfy the following constraints:
\begin{align}
\sum_{e \in F}x_e &\leq f(F)\,, && \text{for all } F \subseteq \adj{a}\,, 
\label{cons:qc-f} \\
x_e & \geq 0\,, && \text{for all } e \in \adj{a}\,. \nonumber
\end{align}
Notice that these constraints are only a subset of the constraints of 
$\qclp$.

Let $\qcpt{a}$ denote the polytope defined by the above constrains.  The extreme
points of $\qcpt{a}$ have a nice structure that becomes crucial when designing
the probability distribution $\qcdist{a}$ over permutations of edges.  Namely,
for any extreme point, the sets for which Constraint~\eqref{cons:qc-f} is tight
form a chain. Moreover, each set in the chain has exactly one more element than
its predecessor and this element is non-zero coordinate of the extreme point.
Formally, we have Lemma~\ref{lem:qc-chain} below.

\begin{lemma} \label{lem:qc-chain}
Let $\mathbf{y} = (y_e)_{e \in \adj{a}}$ be an extreme point of $\qcpt{a}$ 
and let $Y = \{e \in \adj{a} : y_e > 0\}$ be the set of edges that correspond 
to the non-zero coordinates of $\mathbf{y}$.  
Then there exist $|Y|$ subsets $S_1, \dots, S_{|Y|}$ of $\adj{a}$ such that
$S_1 \subsetneq S_2 \subsetneq \dots \subsetneq S_{|Y|}$ with the following 
properties:
\begin{enumerate}
\item Constraint~\eqref{cons:qc-f} is tight for all $S_1, S_2, 
\dots S_{|Y|}$. 
That is $\sum_{e\in S_i} y_e = f(S_i)$ for all $i = 1, \dots, |Y|$.
\item For each $i = 1, \dots, |Y|$, the set $S_i \setminus 
S_{i-1}$ contains exactly one element $e_i$, and $y_{e_i}$ is non-zero
(i.e., $e_i \in Y$).
\end{enumerate}
\end{lemma}

\begin{proof}
It is clear that at least $|Y|$ constraints in~\eqref{cons:qc-f} are tight.
If $|Y| = 1$, the claim of the lemma is obviously true. 
Suppose that $|Y| > 1$.
Now let $A, B \subseteq \adj{a}$ be two different sets for which 
Constraint~\eqref{cons:qc-f} is tight. 
Then we have 
$$f(A) + f(B) = \sum_{e \in A} y_e + \sum_{e \in B} y_e 
	= \sum_{e \in A \cup B} y_e+ \sum_{e \in A \cap B} y_e 
	\leq f(A \cup B) + f (A \cap B),$$
where the last inequality follows because $\mathbf{y}$ satisfies 
Constraint~\eqref{cons:qc-f}. 

Observe that, if $A \nsubseteq B$ and $B \nsubseteq A$, then by 
Lemma~\ref{lem:qc-strict-submod}, $f(A) + f(B) > f(A \cup B) + f(A \cap B)$. 
Thus, it must be the case that either $A \subsetneq B$ or $B \subsetneq A$, 
and consequently there exist $|Y|$ sets $S_1, S_2, \dots, S_{|Y|}$ such that 
$S_1 \subsetneq S_2 \subsetneq \dots \subsetneq S_{|Y|}$, for which 
Constraint~\eqref{cons:qc-f} is tight. 

For each $i = 1, 2, \dots, |Y|$, we thus have that $\sum_{e \in S_i 
\setminus S_{i-1}}y_e = f(S_i) - f(S_{i-1}) > 0$, where the 
inequality is due to the strictly increasing property of $f$. 
This implies that each $S_i \setminus S_{i-1}$ must contain at least one 
edge $e_i$ such that $y_{e_i} >  0$, and since there are only $|Y|$ 
non-zero coordinates in $\mathbf{y}$, each $S_i \setminus S_{i-1}$ must
contain exactly one such $e_i$.  Now, suppose that some $S_i \setminus S_{i-1}$
contains some $e'_i$ such that $y_{e'_i} = 0$. 
Then $f(S_i) = \sum_{e \in S_i} y_e = \sum_{e \in S_i \setminus \{e'_i\}} 
y_e \leq f(S_i \setminus \{e'_i\}) < f(S_i)$ yields a contradiction. 
Here, the first inequality is due to Constraint~\eqref{cons:qc-f} whereas
the last inequality is due the strictly increasing property of $f$.
\end{proof}

Now fix a vertex $a \in A$ and consider the simple query algorithm given in 
Algorithm~\ref{alg:qc-query}, which outputs at most one edge adjacent to 
vertex $a$.
In Algorithm~\ref{alg:qc-query}, $\qcdist{a}$ is a distribution over
the permutations of edges in some subsets of $\adj{a}$ that,
by~\Cref{lem:qc-distrbutions}, can be found in polynomial time.
We have the following lemma considering Algorithm~\ref{alg:qc-query}.

\begin{algorithm}[H]
	\DontPrintSemicolon
	Draw a permutation $\sigma$ from $\qcdist{a}$ for some fixed vertex 
	$a \in A$. 
	\label{algstep:select-permutation-1} \; 
	\ForEach{edge $e$ in the order of $\sigma$}{ 
		Query edge $e$ to check whether it exists. \;
		If edge $e$ exists, output $e$ and terminate. \;
		}
	\caption{Query algorithm for selecting an edge adjacent to a fixed vertex 
		$a \in A$.}
  \label{alg:qc-query}
\end{algorithm}

\begin{lemma} \label{lem:qc-distrbutions}
	Let $\mathbf{x}^\ast$ be an optimal solution to $\qclp$ and let 
	$\mathbf{x}^\ast_a = (x^\ast_e)_{e \in \adj{a}}$ be its restriction to 
	the coordinates that corresponds to edges in $\adj{a}$.
	Then there exists a distribution $\qcdist{a}$ over the permutations of 	
	subsets of $\adj{a}$ with the following property:
	When the permutation $\sigma$ is drawn from $\qcdist{a}$ in 
	Algorithm~\ref{alg:qc-query}, the probability 
that the algorithm outputs the edge $e$ is $x^\ast_e$, hence the expected	weight of the edge output by 
	Algorithm~\ref{alg:qc-query} is $\sum_{e \in E_a}x^\ast_e w(e)$.
	Moreover, a permutation of edges from $\qcdist{a}$ can be sampled in 
	polynomial time.
\end{lemma}

\begin{proof}

Let $\mathbf{y} = (y_e)_{e \in \adj{a}}$ be an extreme point of 
$\qcpt{a}$ and let $|Y|$ be set of non-zero coordinates of $\mathbf{y}$.
Let $\emptyset = S_0 \subsetneq S_1 \subsetneq \dots \subsetneq S_{|Y|}$
be the chain of sets (for which Constraint~\eqref{cons:qc-f} is tight) 
guaranteed by Lemma~\ref{lem:qc-chain} for the extreme point $\mathbf{y}$. 
We can efficiently find the chain by first setting $S_{{Y}} = Y$, and 
iteratively recovering $S_{i-1}$ from $S_i$ by trying all possible $S_{i} 
\setminus \{e \}$ for $e \in S_i$ to check whether 
Constraint~\eqref{cons:qc-f} is tight.
For each $i = 1, \dots, |Y|$, let $e_i$ be the unique element in 
$S_{i} \setminus S_{i-1}$ and let $\sigma_y = (e_1, \dots, e_{|Y|})$. 
Notice that $S_i = \{e_1, \dots, e_i\}$, and hence $y_{e_i} = 
\sum_{e \in S_i}y_{e} - \sum_{e \in S_{i-1}} y_{e} = f(S_i) - f(S_{i-1})$.
Thus if we select $\sigma_y$ as the permutation in 
Algorithm~\ref{alg:qc-query} and query according to that order, the 
probability that it outputs the edge $e_i$ is exactly 
$$\Pr[\text{some edge in $S_i$ appears}] - \Pr[\text{some edge in $S_{i-1}$
  appears}] =f(S_i) - f(S_{i-1}) = y_{e_i}.$$

Note that the point $\mathbf{x}^\ast_a$ is contained in polytope $\qcpt{a}$.
Thus, using the constructive version of Caratheodary's theorem,  we can 
efficiently find a convex combination $\mathbf{x}^\ast_a = \sum_{i \in [k]}a_i 
\cdot \mathbf{y}^{(i)}$, where $a_i \geq 0$ for all $i \in [k]$, 
$\mathbf{y}^{(i)}$ is an extreme points of $\qcpt{a}$ for all $i \in [k]$, 
$\sum_{i \in [k]} a_i = 1$, and $k = \poly(|E_a|)$. 
This is because we can optimize a linear function over $\qcpt{a}$ in 
polynomial time using submodular minimization as a separation oracle,
and for such polytopes, the constructive version of Caratheodary's
theorem holds (See Theorem~6.5.11 of \cite{Grotschel1988}). 

Define the distribution $\qcdist{a}$ such that it gives permutation 
$\sigma_{\mathbf{y}^{(i)}}$  with probability $a_i$. 
If follows that, if we sample according to this distribution in 
Algorithm~\ref{alg:qc-query}, then for any fixed edge $e$, the probability 
that the algorithm outputs the edge $e$ is $\sum_{i \in [k]}a_i \cdot y^{(i)}_e  
= x^\ast_e$. 
Consequently, the expected weight of the output of 
Algorithm~\ref{alg:qc-query} is $\sum_{e \in E_a} x^\ast_e \cdot w(e)$.
\end{proof}

\subsection{Proposed algorithm and analysis}
\label{sec:qc-algo}

Suppose that we run Algorithm~\ref{alg:qc-query} for all vertices 
$a \in A$ and let $M'$ be the set of all output edges.
Then we have $$\EE \left[\sum_{e \in M'} w(e) \right] = \sum_{a \in A} 
\sum_{e \in \adj{a}} x^\ast_e \cdot w(e) = \sum_{e \in E} x^\ast_e \cdot w(e) 
\geq \opt.$$
Furthermore, $M'$ contains at most one adjacent edge per each vertex 
$a \in A$.
However $M'$ may contain more than one adjacent edge for some vertices 
$b \in B$, and hence it may not be valid matching.

Now again suppose that we run Algorithm~\ref{alg:qc-query} as described 
above for all vertices $a \in A$ in some arbitrary order. 
Consider some fixed vertex $b \in B$.
By~\Cref{lem:qc-distrbutions}, from the perspective of $b$, an edge $e \in \adj{b}$ appears with probability 
$x^\ast_e$ (when we say an edge $e = (a,b)$ appears, it means that 
Algorithm~\ref{alg:qc-query}, when run on vertex $a$, outputs the edge $e$).
Viewing the vertices in $A$ as buyers, we think of the appearance of an
edge $e = (a, b)$ as a buyer $a$ making a take-it or leave-it offer of 
value $w_e$ for item $b$.
Thus if we use a uniformly random order of vertices in $A$, picking an 
edge adjacent to the fixed vertex $b$ can be viewed as an instance of the 
prophet secretary problem.

The $(1 - 1/e)$-competitive algorithm algorithm given by 
Ehsani et al.~\cite{Ehsani2018} for the prophet secretary problem first sets 
a base price for the item. 
If some buyer comes at time $t \in [0, 1]$, and if the item is not already
sold, then the algorithm sells the item to this buyer if the offered price
is at least $(1 - e^{t - 1})$ times the base price.
Since the prophet secretary problem deals with a single item,
the goal is to choose the buyer with highest offer, and hence
they set base price of the item as the expected value of the maximum
offer.

However, rather than picking the maximum weighted edge adjacent to each $b$,
we want to maximize the total weight of the matching constructed. 
Thus, we set the base price $c_b$ for each $b$, not as the the expectation
of the offline secretary problem, but as the expected weight of the
edge adjacent to $b$ in some optimal offline maximum-weight bipartite 
matching.
To be concrete, we set $c_b = \sum_{e \in \adj{b}}x^\ast_e \cdot w(e)$ 
(recall that we can think of $x^\ast_e$ as the probability that some fixed 
optimal algorithm for maximum weighted bipartite matching in \qc 
model adds edge $e$ to its output).

\begin{algorithm}[!htb]
	\DontPrintSemicolon
	Solve $\qclp$ to get $\mathbf{x}^\ast$ and find the permutation 
		distributions 	$\qcdist{a}$ for all $a \in A$.\;
	For each vertex $a \in A$, select $t_a \in [0, 1]$ (arrival time) 
		independently and uniformly at random.\;
	For each vertex $b \in B$, set the base price $c_b = \sum_{e \in E_b} 
		x^\ast_e \cdot w(e).$\; 
	Let $M$ be an empty matching. \;
	\ForEach{vertex $a \in A$ in the increasing order of $t_a$}{
		Draw a permutation $\sigma$ of edges from $\qcdist{a}$. 
			\label{algstep:qc-draw-perm}\; 
		\ForEach{$e = (a,b)$ in the order of $\sigma$}{ 
			\label{algstep:qc-edge}
			\If{$w(e) \geq (1 - e^{t_a - 1}) \cdot c_b$ {\bf and } $b$ is not 	
				matched}{ 
				Query edge $e$ to check whether it exists. 
					\label{algstep:qc-query} \;
				If it exists, add it to $M$ and continue to next vertex in $A$. 
					\label{algstep:qc-pick}\;
				}
		  \Else{
		  		Flip a coin that give \textsc{Heads} with probability $p_e$. 
		  			\label{algstep:qc-simquery}\;
		  		If \textsc{Heads}, continue to next vertex in $A$.
					\label{algstep:qc-simpick}\; 
			} 
		}
	}
	\KwRet the matching $M$. \;
  \caption{Outline of \approxqc.}
  \label{alg:qc-matching}
\end{algorithm}

We present the pseudo-code of our algorithm \approxqc in
Algorithm~\ref{alg:qc-matching}.
We start by independently assigning each $a \in A$ a uniformly random
arrival time $t_a \in [0,1]$, and then for each vertex $a \in A$ in 
the order of the arrival time, we run a slightly modified version of the 
query algorithm given in Algorithm~\ref{alg:qc-query}.
For each $b \in B$, we pick an edge $e = (a,b)$ if it appears and if its
weight exceeds the threshold $(1 - e^{t_a - 1}) \cdot c_b$.
Since a \qc algorithm is committed to adding any queried edge that
exists, we query an edge $e$ only if $e \cap B$ is not already assigned to 
some other vertex $a \in A$ and its weight $w(e)$ exceeds the threshold.
But we still need to make sure that, for a fixed $b$, edges 
$e \in \adj{b}$ appears (in the sense that if we run 
Algorithm~\ref{alg:qc-query}, it outputs the edge $e$) with probability 
$x^\ast_e$.
Hence we have the {\bf else} clause of the conditional in 
Algorithm~\ref{alg:qc-matching} that simulates the behavior of 
Algorithm~\ref{alg:qc-query} in the cases we decide not to actually query 
an edge.

We conclude this section with Theorem~\ref{thm:qc-alg} which shows that
our algorithm \approxqc is $(1 - 1/e)$-approximate.
The proof follows exactly the same lines (except for the definition of
base price $c_b$) as in Ehsani et al.~\cite{Ehsani2018} to show that the 
expected weight of the edge adjacent to a fixed vertex $b \in B$ in the
output of \approxqc is at least $(1 - 1/e) \cdot c_b$.
Then by the linearity of expectation, the expected utility of \approxqc is
at least $(1 - 1/e)  \cdot \sum_{b \in B} c_b = (1 - 1/e) \cdot \sum_{e \in E}
x^\ast_e \cdot w(e) \geq \opt$ (recall that $c_b = \sum_{e \in \adj{b}} x^
\ast_e \cdot w(e)$). 

\begin{theorem} \label{thm:qc-alg}
The expected utility $\qcutil(\approxqc)$ of the algorithm \approxqc is at 
least $(1 - 1/e) \cdot \opt$.
\end{theorem}

\begin{remark}
For the sake of completeness, we reproduce the analysis of Ehsani et al.~\cite{Ehsani2018} in 
Section~\ref{sec:poi} for our more general algorithm in the price of information model (see Theorem~\ref{thm:poiutility}).
\end{remark}

\section{Extension to the Price of Information Model}
\label{sec:poi}

In this section, we present a $(1 - 1/e)$-approximation algorithm for the 
maximum weight bipartite matching problem (\mwbm) in the price of information 
(PoI) model introduced by Singla~\cite{Singla2018}.
Our strategy is essentially the same as that used in Section~\ref{sec:qc} 
for the \qc model except for a few enhancements.

Let $G = (A \cup B, E)$ be a bipartite graph where each edge $e \in E$ 
independently takes some random weight $X_e$ from a known probability
distribution. 
The weight distributions can be different for different edges and are independent.
To find the realization of $X_e$ (i.e., the actual weight of the edge $e$), we 
have to \emph{query} the edge $e$ at a cost of $\pi_e$.  
Consider an algorithm $\mathcal{A}$ that queries a subset $Q$ of edges $E$ and
outputs a valid matching $M \subseteq Q$. 
We call such an algorithm a PoI algorithm for \mwbm.
We define the expected PoI utility of Algorithm~$\mathcal{A}$ as
\begin{align*}
\poiutil(\calA) := \E \left[\sum_{e \in M} X_e - \sum_{e \in Q} \pi_e\right],
\end{align*}
where the expectation is taken over the randomness of $X_e$'s and any
internal randomness of the algorithm.
The goal is to design a polynomial-time PoI algorithm $\calA$ for \mwbm 
such that its expected PoI utility $\poiutil(\calA)$ is maximized. 
Let $\opt = \max \, \poiutil(\mathcal{A})$ denote the optimal expected
PoI utility, where the maximization is over all PoI algorithms.

Let $E=\{e_1, \dots, e_m\}$, and let $\mathbf{X} = (X_{e_1}, \dots, X_{e_m})$.
Let $\calM$ be the collection of all valid bipartite matchings in $G$.  
The following lemma is due to Singla \cite{Singla2018}.

\begin{lemma} \label{lem:poi-ub}
For each edge $e \in E$, let $\tau_e$ be the solution to the equation
$\E[\max\{(X_e - \tau_e),0\}] = \pi_e$ and let $Y_e = \min(X_e, \tau_e)$.
Then the optimal expected PoI utility $\opt$ is upper bounded by
$\E_{\mathbf{X}}\left[ \max_{M \in \calM} \sum_{e \in M}Y_e \right]$.
\end{lemma}

To derive our algorithm, we go through the same two stages as in 
Section~\ref{sec:qc}.
We first construct a linear program (LP), this times defining the constraints
using the probability distributions of $Y_e$'s (that were defined in 
Lemma~\ref{lem:poi-ub}), and use its value together with 
Lemma~\ref{lem:poi-ub} to upper bound $\opt$.
For this, we discretize the distributions of $Y_e$'s, and in contrast to the 
\qc setting, we now define variables $x_{e,v}$ for each edge-value pair 
$(e,v)$; 
We think of $x_{e,v}$ as the joint probability that $Y_e = v$ and some fixed 
optimal PoI algorithm includes edge $e$ in its output.
The objective value of this LP upper bounds the quantity
$\E_{\mathbf{X}}\left[ \max_{M \in \calM} \sum_{e \in M}Y_e \right]$, which
in turn is an upper bound of the optimal expected PoI utility as stated in
Lemma~\ref{lem:poi-ub}.
We describe the construction of our LP in Section~\ref{sec:poi-ub}.
Next, in Section~\ref{sec:poi-structure}, we analyze the structure of our new 
LP as we did in the previous section and use it to define analogous 
probability distributions over subsets of edge-value pairs. 
Finally, in Section~\ref{sec:poi-alg} we put everything together to construct 
our $(1 - 1/e)$-approximate PoI algorithm for \mwbm.

\subsection{Upper-bounding the optimal expected utility.}
\label{sec:poi-ub}

Assume that the distributions of $Y_e$ are discrete\footnote{We can e.g.
  achieve this by geometric grouping into polynomially many classes.}.  
For each $e \in E$, let $V_e$ denote the set of possible values of $Y_e$.  
For each vertex $u \in A \cup B$, let $\ev{u} = \{(e, v) : e \in \delta(u), 
v \in V_e\}$ be the set of all edge-value pairs for all edges incident to $u$. 
Let $\allev = \cup_{u \in A} \ev{u}$ be the set of all edge-value pairs.  
For each edge $e \in E$ and value $v \in V_e$, let $p_{e,v}$ be the 
probability that $Y_e = v$, and for a set $F \subseteq \allev$, let $f(F)$ 
be the probability that $Y_e = v$ for at least one edge-value pair 
$(e,v) \in F$.


Fix any PoI algorithm $\calA$ for \mwbm.
For each edge-value pair $(e,v) \in \allev$, let $A_{e,v}$ be the event 
that $Y_e = v$ and $\calA$ includes edge $e$ in its output, and let 
$x_{e,v} = \Pr[A_{e,v}]$.  
Now fix a vertex $u \in A \cup B$ and a set $F \subseteq \ev{u}$.  
Then $\Pr[ \cup_{(e,v) \in F} A_{e,v}] \leq \Pr[ Y_e = v \text{ for some } 
(e,v) \in F] = f(F)$.  
But since all events $A_{e,v}$ for $(e,v) \in \ev{u}$ are mutually disjoint
(as with the \qc setting, the algorithm $\mathcal{A}$ outputs a valid 
matching and thus the output has at most one edge incident to vertex $u$),
$\Pr[\cup_{(e,v) \in F} A_{e,v}] = \sum_{(e,v) \in F} \Pr[A_{e,v}] = 
\sum_{(e,v)  \in F} x_{e,v}$.  
Thus we have $\sum_{(e,v) \in F} x_{e,v} \leq f(F)$, and this must be true 
for all $u \in A \cup B$ and $F \subseteq \ev{u}$.

Now consider the following LP, which we call $\poilp$.
\begin{align*}
&\text{Maximize } & \sum_{(e,v) \in \allev} x_{e,v} \cdot v,  & & & \\
&\text{subject to } & \sum_{(e,v) \in F} x_{e,v} & \leq f(F) & 
   \text{ for all } F \subseteq \ev{u} \text{ for all } u \in A \cup B,  & \\
&	& x_{e,v} & \geq 0 & \text{for all } (u,v) \in \allev.&
\end{align*}

We have the following lemma concerning $\poilp$.

\begin{lemma} \label{lem:poi-lp-ub}
The optimal expected PoI utility $\opt$ is upper bounded by the value 
of $\poilp$.
\end{lemma}

\begin{proof}
By Lemma~\ref{lem:poi-ub}, we have that $\opt \leq \E_{\mathbf{X}}
[\max_{M \in \mathcal{M}} \sum_{e \in M}Y_e]$.
Now consider an algorithm $\mathcal{A}$ that queries $Y_e$ for all edges 
$e \in E$ and outputs a maximum weighted bipartite matching $M$ of $G$.
Setting $x_{e,v}$ to be the joint probability that $Y_e = v$ and 
$e \in M$ for each edge-value pair $(e,v) \in \allev$ gives a feasible 
solution to $\poilp$. 
Hence $\poiutil(\mathcal{A}) =  \E_{\mathbf{X}}[\max_{M \in \mathcal{M}} 
\sum_{e \in M}Y_e] \leq \sum_{(e,v) \in \allev} x^\ast_{e,v} \cdot v$, 
where $\mathbf{x}^\ast = (x^\ast_{e,v})_{(e,v) \in \allev}$ is an optimal
solution of $\poilp$.
\end{proof}

\subsection{Structure of the LP and its implications.}
\label{sec:poi-structure}

Now we analyze the structure of $\poilp$. 
Our analysis closely follows that of Section~\ref{sec:qc-lp-structure}.

As usual, $f$ is a coverage function and hence it is submodular.
Thus, for each vertex $u \in A \cup B$, we can use submodular minimization 
to check whether any constraint of the form $\sum_{(e,v) \in F} x_{e,v} \leq 
f(F)$ is violated for any subset $F \subseteq \ev{u}$. 
This yields Lemma~\ref{lem:poi-solve-lp} below.

\begin{lemma} \label{lem:poi-solve-lp}
The linear program $\poilp$ is solvable in polynomial-time.
\end{lemma}

We proceed as follows.
Consider the query strategy given in Algorithm~\ref{alg:poi-query} that
queries edges incident to a fixed vertex $a \in A$ in some random order.
This is the PoI version of the Algorithm~\ref{alg:qc-query} given for the
\qc setting.
Following (almost) the  same procedure as in 
Section~\ref{sec:qc-lp-structure}, we find distributions $\poidist{a}$ 
that makes Algorithm~\ref{alg:poi-query} pick an edge $e$ that has value $v$ 
with probability $x^\ast_{e,v}$, and then use those to construct a PoI 
algorithm  for \mwbm that gives $(1 - 1/e)$ approximation guarantee.
But the issue here is that Algorithm~\ref{alg:poi-query} considers 
the distributions of $Y_e$'s and does not pay query costs whereas our final
approximate PoI algorithm needs to consider the distributions of $X_e$'s and 
has to pay query costs.

\begin{algorithm}[H]
	\DontPrintSemicolon
	Let $z_e = \operatorname{Null}$ for all $e \in \adj{a}$.\;
	Draw a permutation $\sigma$ from $\poidist{a}$. 
		\label{algstep:poi-select-permutation} \; 
	\ForEach{$(e, v)$ in the order of $\sigma$}{ 
	If $z_e = \operatorname{Null}$, draw $z_e$ from a distribution identical 
		to that of $Y_e$.\; 
	If $z_e = v$, output $e$ and terminate.
	}
	\caption{Query algorithm for selecting an edge incident to a fixed vertex 
		$a \in A$.}
  \label{alg:poi-query}
\end{algorithm}

Now consider the way we defined $\tau_e$ (which we used to define 
$Y_e$'s), and observe that the values of $X_e$ above the threshold 
$\tau_e$, on expectation, covers the cost $\pi_e$ of querying it.
Thus, if we can make sure that the first time we query an edge $e$ (i.e.,
the time where we pay the price $\pi_e$) in 
Algorithm~\ref{alg:poi-query} is for the value $\tau_e$, then
we can still use it to construct our final PoI matching algorithm
(where we actually query $X_e$ values, and when an edge $e$ is queried 
for the first time for value $\tau_e$, in expectation we actually get
a net value of $\tau_e$ with probability $x^\ast_{e, \tau_e}$ after paying 
$\pi_e$).
Using a careful construction, we make sure that distributions $\poidist{a}$ 
only gives those permutations where for any edge $e$, 
the pair $(e, \tau_e)$ appears before any other pair $(e, v)$.

For such a construction, we consider a slightly different polytope $\poipt{a}$ 
(as opposed to how we defined $\qcpt{a}$) for each $a \in A$.
Fix a vertex $a \in A$, and consider the family $\lev{a}$ of subsets of 
$\ev{a}$ defined as follows:
$$\lev{a} := \{F \subseteq \ev{a}: (e,v) \in F \Rightarrow (e,v') \in F 
\text{ for all } v' \geq v \text{ such that } (e,v') \in \ev{a} \}.$$ 
I.e., $\lev{a}$ is a family of subsets of $\ev{a}$ that satisfy the 
following: If a set $F$ of edge-value pairs is in $\lev{a}$ and an edge-value 
pair $(e,v)$ is in $F$, then $F$ also contains all edge-value pairs for the 
same edge $e$ having values greater than $v$.  
It is easy to verify that if $A, B \in \lev{a}$ then both $A \cup B \in 
\lev{a}$ and $A \cap B \in \lev{a}$, which makes $\lev{a}$ a lattice family.
(I.e., the sets in $\lev{a}$ forms a lattice where intersection and union 
serve as \emph{meet} and \emph{join} operations respectively.) 

We define below the polytope $\poipt{a}$ using a constraint for each set in the
family $\lev{a}$. 
\begin{align}
\sum_{(e,v) \in F} x_{e,v} & \leq f(F)  & \text{ for all } F \in \lev{a} 
 \label{cons:poi-f}\\
x_{e,v} & \geq 0 & \text{for all } (e,v) \in \ev{a}. \nonumber
\end{align}

Analogous to our assumption $0 < p_e < 1$ for all $e \in E$ for the 
\qc setting, we now assume that each $p_{e,v} > 0$ and
for each edge $e$, $\sum_{v \in V_e} p_{e,v} < 1$. 
(I.e., we can assume that with some small probability $p_{e,\star}$ the 
edge $e$ does not exist or equivalently, we can also assume $p_{e,0} > 0$ 
and $0 \notin V_e$.
We omit the details, but one can use the same argument of re-scaling the 
probabilities to justify this assumption.)
Under these assumptions on $p_{e,v}$'s, we have the following lemma.
The proof resembles that of Lemma~\ref{lem:qc-strict-submod} from
Section~\ref{sec:qc-lp-structure}, and we defer it to 
Appendix~\ref{app:proofs}.

\begin{lemma} \label{lem:poi-strict-submod}
Fix a vertex $a \in A$.
If $p_{e,v} > 0$ for all $(e,v) \in \ev{a}$ and
$\sum_{v \in V_e} p_{e,v} < 1$ for all $e \in \adj{a}$, the function $f$ 
is strictly submodular and strictly increasing on the lattice family 
$\lev{a}$.
Formally,
\begin{enumerate}
\item \label{prop:a1} For any $A, B \in \lev{a}$ such that $A \setminus B 
\neq \emptyset$ and $B \setminus A \neq \emptyset$, $f(A) + f(B) > f(A \cap B) 
+ f(A \cup B)$, and
\item \label{prop:a2} For any $A \subsetneq B \subseteq \ev{a}$, 
$f(B) > f(A)$.
\end{enumerate}
\end{lemma}

Similarly to the \qc setting, we now analyze the structure of the 
extreme points of polytope $\poipt{a}$.
We have the following lemma, which is a slightly different version of 
Lemma~\ref{lem:qc-chain} from Section~\ref{sec:qc-lp-structure}.

\begin{lemma}
\label{lem:poi-chain} 
Let $\mathbf{y}=(y_{e,v})_{(e,v) \in \ev{a}}$ be an extreme point of 
$\poipt{a}$ and let $Y = \{e \in \ev{a} : y_{e,v} > 0\}$ be the set of 
non-zero coordinates of $\mathbf{y}$.
Then there exist $|Y|$ subsets $S_1, \dots, S_{|Y|}$ of $\ev{a}$  such that
$S_1 \subsetneq S_2 \subsetneq \dots \subsetneq S_{|Y|}$ with the following 
properties:
\begin{enumerate}
\item \label{prop:b1} Constraint~\eqref{cons:poi-f} is tight for all 
$S_1, S_2, \dots S_{|Y|}$. That is $\sum_{(e,v) \in S_i} y_{e,v} = f(S_i)$ 
for all $i = 1, \dots, |Y|$.
\item \label{prop:b2} For each $i = 1, \dots, |Y|$, the set
$(S_i \setminus S_{i-1}) \cap Y$ contains exactly one element $(e_i, v_i)$.
Moreover, for any other $(e, w) \in S_i \setminus S_{i-1}$, we have $e = e_i$
and $w \geq v_i$.
\end{enumerate} 
\end{lemma}

\begin{proof}
Property~\ref{prop:b1} and the fact that each $(S_{i} \setminus S_{i-1}) 
\cap Y$ contains exactly one pair $(e_i, v_i)$ follows from the proof of 
Lemma~\ref{lem:qc-chain}.
It remains to show that each $S_{i} \setminus S_{i-1}$ additionally contains
only those edge-value pairs $(e_i, w)$ for which $w \geq v_i$.

Suppose to the contrary that there is some $S_i \setminus S_{i-1}$ that 
contains at least one other pair $(e',v')$ that violates this property.
Let $S'_{i} = S_{i-1} \cup \{(e_i, w): w \geq v_i\}$.
Then $S_{i-1} \subsetneq S'_{i} \subsetneq S_{i}$ and $S'_i$ is also in 
the family $\lev{a}$.
Thus we have that $f(S'_i) \geq \sum_{(e,v) \in S'_i}y_{e,v} = \sum_{(e,v) 
\in S_i}y_{e,v} = f(S_i),$ which is a contradiction because $f$ is strictly 
increasing and $S'_i \subsetneq S_i$.
Here the first inequality holds because $\mathbf{y}$ is in $\poipt{a}$ 
and the first equality holds because $S'_i$ contains all coordinates in $S_i$
for which $\mathbf{y}$ is non-zero. 
The last equality is true because $S_i$ corresponds to a tight constraint for 
the extreme point $\mathbf{y}$.
\end{proof}

As in the previous section, we are now ready to construct the distribution
$\poidist{a}$.
We present this explicit construction in the proof of 
Lemma~\ref{lem:poi-distributions} stated below, which is the counterpart of 
Lemma~\ref{lem:qc-distrbutions}. 

\begin{lemma} \label{lem:poi-distributions}
Let $\mathbf{x}^\ast$ be an optimal solution of $\poilp$.
For each vertex $a \in A$, there exist a distribution $\poidist{a}$ over the
permutations of (subsets of) edge-value pairs in $\ev{a}$ that 
satisfies the following properties:
\begin{enumerate}
\item \label{prop:c1}
  For each permutation $\sigma$ drawn from $\poidist{a}$, if edge-value 
  pair $(e,v)$ appears in $\sigma$, then the edge-value pair $(e, w)$ 
  appears before $(e,v)$ in $\sigma$ for all $w \in V_e$ 
  such that $w \geq v$,
\item \label{prop:c2}
  $\Pr[\text{Algorithm~\ref{alg:poi-query} outputs $e$}] = 
  \sum_{v \in V_e} x^\ast_{e,v}$ for all $e \in \adj{a}$, and
\item \label{prop:c3}
  $\sum_{v \in V_e: v \geq w} \Pr[\text{Algorithm~\ref{alg:poi-query} 
  	outputs $e$ with value $v$}]\cdot v \geq \sum_{v \in V_e: v \geq w} 
  	x^\ast_{e,v} \cdot v$ for all $e \in \adj{a}$ and $w \in \R^+$.
\end{enumerate}
Moreover, a permutation  of edge-value pairs from $\poidist{a}$ can be sampled 
in polynomial time.
\end{lemma}

\begin{proof}

As with the case of Lemma~\ref{lem:qc-distrbutions}, we first associate a
permutation of edge-value pairs with each extreme point of $\poipt{a}$.

For an extreme point $\mathbf{y}$, let $Y$ and $S_1,\ldots,S_{|Y|}$ be as defined
in~\Cref{lem:poi-chain}.  Consider the permutation $\sigma_y$ of elements in $S_{|Y|}$ that is defined
as follows: Start with $\sigma_y = [\,]$ and for each $i = 1, \dots, |Y|$,
append to it the edge-value pairs in $S_{i} \setminus S_{i-1}$ in the 
decreasing order of value. 
Recall that for all $i = 1, \dots, |Y|$, all edge-value pairs in $S_{i} 
\setminus S_{i-1}$ corresponds to a single edge.
Let $(e,v) \in S_{i}$. 
Then, by the definition of the family $\lev{u}$, $(e, v') \in S_{i}$ for all
$v' \in V_e$ such that $v' \geq v$.  
Thus none of the sets $S_{i+1} \setminus S_{i}, S_{i+2} \setminus S_{i+1},
\dots, S_{|Y|} \setminus S_{|Y| - 1}$ can contain an edge-value pair $(e, v')$ 
such that $v' > v$.  
Also, since the elements in $S_{i} \setminus S_{i-1}$ are appended to 
$\sigma_y$ in decreasing order of values, $\sigma_y$ has the following 
property: If at any point the edge-value pair $(e,v)$ appears in $\sigma_y$, 
then $(e,w)$ appears in $\sigma_y$ before $(e,v)$ for all $w \in V_e$ such 
that $w > v$.

Let $\sigma_y = (e_1, v_1), \dots, (e_\ell, v_\ell)$ be the permutation of 
edge value pairs associated with the extreme-point $\mathbf{y}$.  
Let $T_i := \{(e_1, v_1), \dots, (e_i, v_i) \}$ denote the set of first $i$
edge-value pairs in $\sigma_y$.  
If we select permutation $\sigma_y$ in 
Line~\ref{algstep:poi-select-permutation} in Algorithm~\ref{alg:poi-query}, 
the probability of it picking edge $e_i$ with value $v_i$ is exactly
$f(T_i) - f(T_{i-1})$.  
Now define a new vector $\mathbf{y'}$ with the same indices as $\mathbf{y}$ 
as follows: For each $(e_j, v_j) \in \sigma_y$, $y'_{e_j,v_j} 
= f(T_j) - f(T_{j-1})$, and all the other coordinates of $\mathbf{y'}$ 
are $0$.
Notice that $\mathbf{y'}$ and $\mathbf{y}$ satisfy the following:
\begin{enumerate}
\item \label{prop:d1} $\sum_{v \in V_e} y'_{e,v} = \sum_{v \in V_e} y_{e,v}$ 
for all $e \in \adj{a}$, and 
\item \label{prop:d2} $\sum_{v \in V_e: v \geq w} y'_{e,v} \cdot v \geq 
\sum_{v \in V_e: v \geq w} y_{e,v} \cdot v$ for all $e \in \adj{a}$ and 
$w \in \R^+$. 
\end{enumerate}
To see this fix some set $S_i$ and let $(e_{j},v_{j}), (e_{j+1}, v_{j+1}), 
\dots, (e_{k}, v_{k})$ be all the edge-value pairs in $S_{i} - S_{i-1}$.  
Then we know that $e_j = e_{j+1} = \dots = e_{k}$, $v_j > v_{j+1} > \cdots 
> v_{k}$, and $y_{e_k, v_k} \neq 0$.  
We thus have that $\sum_{j' = j}^k y'_{e_k,v_j} = f(T_{k}) - f(T_{j-1}) 
= f(S_i) - f(S_{i-1}) = y_{e_{k}, v_{k}}$ (recall that $(e_k, v_k)$ is the 
unique element in $S_i \setminus S_{i-1}$ for which $y_{e_k, v_k}$ is 
non-zero).  
This holds for elements in $S_i \setminus S_{i-1}$ in all $i = 1, \dots, |Y|$, 
and yields the Property~\ref{prop:d1} above.  
To see Property~\ref{prop:d2}, notice that within each $S_i \setminus 
S_{i-1}$, the weight of the non-zero coordinate $y_{e_k,v_k}$ is 
re-distributed among $y_{e_j,v_j}, \dots, y_{e_k,v_k}$, and that 
$v_{j'} \geq v_k$ for $j' = j, j+1, \dots, k$.

Let $\sigma$ be the random variable that denotes the permutation picked by
Algorithm \ref{alg:poi-query}. 
Then we have that, for all $e \in \delta(u)$,
$$\Pr[ \text{Algorithm \ref{alg:poi-query} picks $e$} | \sigma = \sigma_y] 
= \sum_{v \in V_e} y'_{e,v} = \sum_{v \in V_e} y_{e,v},$$ and for all
$(e,w) \in \ev{u}$,
$$\sum_{e, v \in V_e: v \geq w} \Pr[\text{Algorithm~\ref{alg:poi-query} 
picks $e$ with value $v$} | \sigma = \sigma_y] \cdot
v = \sum_{e, v \in V_e: v \geq w} y'_{e,v} \cdot v \geq \sum_{e, v \in V_e: v
\geq w} y_{e,v} \cdot v.$$
 
Now let $\mathbf{x}^\ast_a$ be the restriction of the optimal solution the 
coordinates in $\ev{a}$.
Since the constraints that define $\poipt{a}$ are only a subset of the
constraints of $\poilp$, $\mathbf{x}^\ast_a$ lies in $\poipt{a}$.
If we can optimize a linear function of $\poipt{a}$ in polynomial time,
then we can follow the same lines of the proof of 
Lemma~\ref{lem:qc-distrbutions} and use the constructive version of 
Caratheodary's theorem to find a convex combination $\mathbf{x}^\ast_a = 
\sum_{i \in [k]} a_i \cdot \mathbf{y}^{(i)}$, where $k = \poly(|\ev{a}])$
and for each $i \in [k]$, $\mathbf{y}^{(i)}$ is an extreme point of 
$\poipt{a}$.
However, unlike in the \qc case, the polytope $\poipt{a}$ only has 
constraints for sets of a lattice family, and as a result, we cannot 
use the usual submodular minimization as a separation oracle for $\poilp$.
But luckily, Gr{\"{o}}tschel et al.~\cite{Grotschel1981} showed that we can 
minimize any submodular function over a lattice family in polynomial time.
Thus we can efficiently find such a convex combination.

Once we have the convex combination, the rest is exactly the same as
the \qc setting. The distribution $\poidist{a}$ returns the permutation
$\sigma_{y^{(i)}}$ with probability $a_i$ for all $i \in [k]$.
One can easily verify that Properties~\ref{prop:c1}-\ref{prop:c3} hold for 
this distribution.
\end{proof}

\subsection{Proposed algorithm and analysis.}
\label{sec:poi-alg}

We now present a $(1 - 1/e)$-approximate PoI algorithm \approxpoi for \mwbm.

The algorithm closely resembles the algorithm \approxqc we presented
for the \qc model, but has two key differences.
First, we now have to pay a price for querying edges. 
But, as we prove later, this price is already taken care of by the way we
defined $Y_e$ variables.
The second difference is that for each edge $e \in E$, we now have multiple 
values to consider, but regardless, with respect to a fixed vertex 
(i.e. an item) $b \in B$, the vertices $a \in A$ can still be viewed as 
buyers; The appearance of multiple edge-value pairs for the same edge can be 
interpreted as a distribution over values that the buyer $a$ offers for the 
item $b$.

The outline of our algorithm is given in Algorithm \ref{alg:poi-matching}.
Note that, we again have the {\bf else} clause in the conditional to make sure 
that we do not change the probability distributions of the appearances of
edge-value pairs even if we decide not to query some edges for certain values.

\begin{algorithm}[H]
	\DontPrintSemicolon
	Solve $LP^\star$ to get $\mathbf{x}^\ast$ and find the permutation 
		distributions $\poidist{a}$ for all $a \in A$.\;
	For each vertex $a \in A$, select $t_a \in [0,1]$ (arrival time) 
		independently and uniformly at random.\;
	For each vertex $b \in B$, let $c_b = \sum_{(e,v) \in \ev{b}} x^\ast_{e,v} 
		\cdot  v.$\; \label{line:poi-matching:cbdef}
	Let $z_e = \operatorname{Null}$ for all $e \in E$. \;
	Let $M$ be an empty matching. \;
	\ForEach{vertex $a \in A$ in the increasing order of $t_a$}{
		Draw a permutation $\sigma$ from $\poidist{a}$. 
			\label{algstep:poi-draw-perm}\; 
		\ForEach{$(e = (a,b), v)$ in the order of $\sigma$}{
			\If{$v \geq (1 - e^{t_a - 1}) \cdot c_b$ {\bf and } 
					$b$ is not matched}{ 
				If $z_e = \operatorname{Null}$, pay $\pi_e$ and query edge $e$ to  
				find its actual value. Let $z_e$ be this value.
					\label{algstep:poi-query}\; 
				If $\min(z_e, \tau_e) = v$, add $e$ to $M$ and continue to next 
					vertex in 
					$A$. \label{algstep:poi-pick}\; 
			} \Else {
				If $z_e = \operatorname{Null}$, draw $z_e$ from a  distribution 
					identical to that of $X_e$. \label{algstep:poi-simquery}\; 
			  If $\min(z_e, \tau_e) = v$,  continue	to next vertex in $A$. 
			 		\label{algstep:poi-simpick}\; 
			} 
		} 
	}
	\KwRet the matching $M$.
	\caption{Outline of \approxpoi.}
	\label{alg:poi-matching}
\end{algorithm}

The analysis of the expected PoI utility of \approxpoi consists of two parts.
First we use the same technique used by Singla~\cite{Singla2018} to show that
we can analyze the expected utility using $Y_e$ variables and no query
costs instead of using $X_e$ variables with query costs. 
Next we reproduce almost the same analyis by Ehsani et al.~\cite{Ehsani2018} 
to prove the approximation guarantee.

Let $Z$ be the value we get from Algorithm~\ref{alg:poi-matching}.  
Then $Z = \sum_{e \in E} \left( \indpick{e} X_e - \indquery{e} \pi_e \right),$
where $\indpick{e}$ and $\indquery{e}$ are the indicator variables for the
events that edge $e$ is picked in Line~\ref{algstep:poi-pick} and it is 
queried in Line~\ref{algstep:poi-query} respectively.

Now suppose that we run an identical copy of Algorithm~\ref{alg:poi-matching}
in parallel but we do not pay for querying in Line~\ref{algstep:poi-query}, 
but instead of gaining $X_e$, we only get $Y_e = \min(X_e, \tau_e)$.  
We say that this latter execution in the ``free-information'' world whereas 
the original algorithm runs in the ``price-of-information'' world or PoI 
world. 
Let $Z'$ be the value we get in the free information world.  
We have the following lemma.

\begin{lemma} \label{lem:poi-vs-free}
The expected utility $\E[Z]$ of \approxpoi (which is the PoI world) is equal 
to the expected utility of its counterpart in the free information world.
\end{lemma}

\begin{proof}
Consider a case where the algorithm picks an edge that is already queried 
before.
In this case, both algorithms get the same value, so the expected increase to
$Z$ and $Z'$ are the same.

Now consider the case where both algorithms query for some edge $e$. 
Notice that if an edge $e$ is queried at any point, it is queried for the 
edge-value pair $(e, \tau_e)$.  
This is because $\tau_e$ is the maximum possible value for an edge $e$, and 
as a result, $(e, \tau_e)$ appears before any other $(e,v)$ for $v < \tau_e$ 
(if it appears at all) in the permutations chosen in 
Line~\ref{algstep:poi-draw-perm} of Algorithm \ref{alg:poi-matching}.  
In this case, the expected increase to $Z'$ in the free information world is
$\tau_e \cdot \Pr[X_e \geq \tau_e]$.  
The expected increase to $Z$ in the PoI world is
\begin{align*}
- \pi_e + \int_{\tau_e}^{\infty} t \cdot p_e(t) dt
&= -\pi_e + \int_{\tau_e}^{\infty} (t - \tau_e) \cdot p_e(t)\, dt +
  \int_{\tau_e}^{\infty} \tau_e\cdot p_e(t) \, dt \\
&= \underbrace{- \pi_e + \E[ (X_e - \tau_e)^+]}_{0} + \tau_e \cdot 
  \Pr[X_e \geq \tau_e].
\end{align*}
\end{proof}

To conclude our analysis, we now present following theorem on the
approximation guarantee.

\begin{theorem} 
  \label{thm:poiutility}
The expected PoI utility $\poiutil(\approxpoi)$ of \approxpoi is at least 
$(1 - 1/e) \cdot \opt$.
\end{theorem}

\begin{proof}
By Lemma~\ref{lem:poi-vs-free}, we have $\poiutil(\approxpoi) = \E[Z] = 
\E[Z']$.
Since the optimal value of $\poilp$ is an upper bound on the 
optimal value $\opt$ (by Lemma~\ref{lem:poi-lp-ub}), it is sufficient to 
show that
$$\E[Z'] \geq (1 - 1/e) \cdot \sum_{(e,v) \in \allev} x^\ast_{e,v} v = 
(1 - 1/e) \cdot \sum_{b \in B} c_b.$$ 
(Recall that $c_b = \sum_{(e,v) \in \ev{b}} x^\ast_{u,v} \cdot v$ as defined
on~\Cref{line:poi-matching:cbdef} in~\Cref{alg:poi-matching}.)
Let $Z'_b = \sum_{a \in \adj{b}} \indpick{a,b} Y_{a,b}$ so that
$Z' = \sum_{b \in B} Z'_b$.  
We show that for any $b \in B$, $\E[Z'_b] \geq (1 - 1/e) \cdot c_b$.  
Then the theorem follows from the linearity of expectation. The following calculations now follow the analysis in~\cite{Ehsani2018} and are included for completeness.

We proceed by splitting $Z'_{a,b}$ into two parts:
$$Z'_{a,b} = \underbrace{\indpick{a,b} (Y_{a,b} - 
(1 - e^{t_a-1}) \cdot c_b)}_{M_{a,b}} + \underbrace{\indpick{a,b} 
(1 - e^{t_a-1}) \cdot c_b}_{N_{a,b}}.$$ 
Define $r(t) := \Pr[\text{no edge incident to $b$ is picked before time 
$t$}]$ and $\alpha(t):= 1 - e^{t-1}$. 
Notice that $r(t)$ is decreasing, and since our $t_a$'s are from a continuous 
distribution, $r(t)$ is a differentiable function. 
Thus we have
\begin{align*}
\E\left[\sum_{a \in \adj{b}} N_{a,b}\right]
&=  - \int_{0}^{1} r'(t) \cdot (1 - 1/e^{t-1})  \cdot c_b \, dt = -  c_b 
  \int_{0}^{1} r'(t) \cdot \alpha(t)\,dt. \\
\intertext{By applying integration by parts,}
\E\left[\sum_{a \in \adj{b}} N_{a,b}\right]
&= - c_b \left( \left[ r(t) \cdot \alpha(t)  \right]_0^1 - \int_0^1 r(t) 
	\cdot \alpha'(t)\,dt \right)\\
&= c_b \left( (1 - 1/e) + \int_0^1 r(t) \cdot \alpha'(t)\,dt \right). 
\addtocounter{equation}{1}\tag{\theequation} \label{eq:rev}
\end{align*}
Now we consider the expectation of $M_{a,b}$.  
Note that the inequality below is due to the third property of the  
distributions $\poidist{a}$ of edge-value pairs we used in the algorithm
(see Lemma~\ref{lem:poi-distributions}).
\begin{align*}
\E[M_{a,b} | t_a = t] &\geq \Pr[\text{no edge incident to $b$ is picked 
before $t$} | t_a = t] 
\sum_{ \substack{v \in V_{a,b} \\ v \geq \alpha(t) \cdot c_b}} 
x^\ast_{(a,b),v} 
(v - \alpha(t) \cdot c_b ).
\end{align*}
 
If $t_a = t$, this means that $t_a$ could not have arrived before $t$.  
Hence
$\Pr[$no edge incident to $b$ is picked before $t$ $ | t_a = t] \geq 
\Pr[$no edge incident to $b$ is picked before $t$ $]$,
and thus we have
\begin{align*}
\sum_{a \in \delta(b)} &\E[M_{a,b} |t_a = t] \\
&\geq \Pr[\text{no edge incident to $b$ is picked before $t$} | t_a = t] 
	\sum_{a \in \adj{b}} \sum_{\substack{v \in V_{a,b} \\ v \geq \alpha(t) 
	\cdot c_b}} 
	x^\ast_{(a,b),v} (v - \alpha(t) \cdot c_b ) \\
&\geq \Pr[\text{no edge incident to $b$ is picked before $t$}] 
 	\sum_{a \in \adj{b}} \sum_{\substack{v \in V_{a,b}\\ v \geq \alpha(t) 
 	\cdot c_b}} 
  x^\ast_{(a,b),v} (v - \alpha(t) \cdot c_b ) \\
&\geq r(t) \sum_{a \in \adj{b}} \sum_{v \in V_{a,b} } x^\ast_{(a,b),v} 
    (v - \alpha(t) \cdot c_b ) \\
&= r(t) \left(c_b -  \alpha(t) \cdot c_b \sum_{a \in \adj{b}} \sum_{v \in 
V_{a,b}}   x^\ast_{(a,b),v} \right) \\
&\geq  r(t) (1 - \alpha(t)) \cdot c_b.
\end{align*}
Since $t_a$ is uniformly distributed over $[0,1]$ for each $a \in A$, 
it follows that
\begin{align}
\E\left[\sum_{a \in \adj{b}} M_{a,b} \right]
& = \int_0^1 \sum_{a \in \adj{b}} \E[ M_{a,b} | t_a = t] \, dt \geq  
  c_b \int_0^1 r(t) \left( 1 - \alpha(t) \right)\, dt. 
  \addtocounter{equation}{1} \tag{\theequation} \label{eq:ut}
\end{align}
Now $\eqref{eq:rev} + \eqref{eq:ut}$ yields
\begin{align*}
\E \left[ \sum_{a \in \adj{b}} Z'_{a,b} \right]
&= \E\left[\sum_{a \in \adj{b}} N_{a,b} \right] + \E \left[ 
	\sum_{a \in \adj{b}} M_{ij} \right] \\
&\geq c_b \left((1 - 1/e) + \int_0^1 r(t) \cdot \alpha'(t)\,dt \right) 
	+ c_b \int_0^1 r(t) \left( 1 - \alpha(t) \right)   \, dt \\
&= c_b  (1 - 1/e) +  c_b  \int_0^1 r(t) \underbrace{\left(1 - \alpha(t) +
  \alpha'(t) \right)}_{0}   \, dt \\
  &= (1 - 1/e) \cdot c_b.
\end{align*}

\end{proof}

\section*{Acknowledgements} We are grateful to Anupam Gupta and Amit Kumar
for influential discussions at an early stage of this work.

\bibliographystyle{alpha}
\bibliography{ref}

\appendix
\section{Proofs of Some Supplementary Results}
\label{app:proofs}

\begin{lemma}[Lemma~\ref{lem:poi-strict-submod}] 
\label{lem:poi-strict-submod-proof}
Fix a vertex $a \in A$.
If $p_{e,v} > 0$ for all $(e,v) \in \ev{a}$ and
$\sum_{v \in V_e} p_{e,v} < 1$ for all $e \in \adj{a}$, the function $f$ 
is strictly submodular and strictly increasing on the lattice family 
$\lev{a}$.
Formally,
\begin{enumerate}
\item \label{prop:aa1} For any $A, B \in \lev{a}$ such that $A \setminus B 
\neq \emptyset$ and $B \setminus A \neq \emptyset$, $f(A) + f(B) > f(A \cap B) 
+ f(A \cup B)$, and
\item \label{prop:aa2} For any $A \subsetneq B \subseteq \ev{a}$, 
$f(B) > f(A)$.
\end{enumerate}
\end{lemma}

\begin{proof}
It is easy to see that  $f(F) = 1 - \prod_{e \in E} (1 - \sum_{v \in V_e : 
(e,v) \in F} p_{e,v}) $.
Consider two sets $A, B \in \lev{a}$ such that $A \setminus B \neq \emptyset$
and $B \setminus A \neq \emptyset$.
Fore each $e \in E$, let $a_e = 1 - \sum_{v \in V_e: (e,v) \in A} p_{e,v}$
and $b_e = 1 - \sum_{v \in V_e: (e,v) \in B} p_{e,v}$.
Thus we have that $f(A) = 1 - \prod_{e \in E}a_e$, 
$f(B) = 1 - \prod_{e \in E} b_e$, $f(A \cup B) = 1 - \prod_{e \in E}
\min(a_e, b_e)$, and $f(A \cap B) = 1 - \prod_{e \in E} \max(a_e, b_e)$.
The last two equations follow from the definition of the family $\lev{a}$.

Now we have
\begin{align}
f(A) + f(B) - f(A \cup B) - f(A \cap B) 
&= \prod_{e \in E} \min(a_e, b_e) + \prod_{e \in E} \max(a_e, b_e)  - 
\prod_{e \in E} a_e - \prod_{e \in E} b_e. \label{eq:diff}
\end{align}
Thus to prove Property~\ref{prop:aa1}, it is sufficient to prove that the 
right hand side of \eqref{eq:diff} is strictly greater than zero,
which is equivalent to showing that
$\prod_{e \in E} \min(a_e, b_e) + \prod_{e \in E} \max(a_e, b_e)  > 
\prod_{e \in E} a_e + \prod_{e \in E} b_e$.

Since $A \setminus B \neq \emptyset$, we have $a_e < b_e$ for at least one 
edge $e_1 \in E$.
To see this, let $(e,v) \in A \setminus B$.
Then for such edge $e$, it follows from the definition of set family
$\lev{a}$ that the set $A_e = A \cap \{(e,w): w \in V_e\}$ contains all the 
elements in the set $B_e = B \cap \{(e,w): w \in V_e\}$, and in addition, 
it also contains at least one more element, namely $(e,v)$. 
Since $p_{e,v} > 0$, $a_e < b_e$.
Similarly, we have $a_e > b_e$ for at least one edge $e_2 \in E$ 
(since we assumed that $p_{e,v} > 0$ for all $(e,v) \in \ev{a}$).
Let $E_1$ be the (nonempty) set of edges for which $a_e < b_e$, and let $E_2$ 
be the (nonempty) set of edges for which $a_e > b_e$.
Without loss of generality we assume that $E_1 \cup E_2 = E$
(if $a_e = b_e$ for some $e$, then we can divide \eqref{eq:diff} by $a_e$
since $\sum_{v \in V_e} p_{e,v} < 1$, $a_e > 0$).
We then have
\begin{align*}
\prod_{e \in E} \min(a_e, b_e) + \prod_{e \in E} \max(a_e, b_e) 
&=  \prod_{e \in E_2} b_e \prod_{e \in E_1} a_e + 
		\prod_{e \in E_2} a_e \prod_{e \in E_1} b_e \\
&> \prod_{e \in E_2} a_e \prod_{e \in E_1} a_e + 
	 \prod_{e \in E_2} b_e \prod_{e \in E_1} b_e \\
&= \prod_{e \in E} a_e + \prod_{e \in E} b_e
\end{align*}
as required. 
The inequality above follows from the rearrangement inequality as 
$\prod_{e \in E_1} b_e > \prod_{e \in E_1} a_e$ and 
$\prod_{e \in E_2} a_e > \prod_{e \in E_2} b_e$.

As for Property~\ref{prop:aa2}, suppose that $A \subsetneq B$.
Then for all $e \in E$, $a_e \geq b_e$, and for at least one $e \in E$,
$a_e > b_e$. 
Hence  $f(B) - f(A) =  \prod_{e \in E} a_e - \prod_{e \in E} b_e > 0$.
\end{proof}

\end{document}